\documentclass[letterpaper]{article}

\usepackage{aaai17}
\usepackage{times}
\usepackage{helvet}
\usepackage{courier}

\usepackage{amsmath}
\usepackage{float}
\usepackage{subfig}
\usepackage[pdftex]{graphicx}
\usepackage{algorithm}
\usepackage{algorithmic}
\usepackage{longtable}
\usepackage{booktabs}
\usepackage{amsthm}
\newtheorem{theorem}{Theorem}

\newtheorem{definition}{Definition}

\newcommand*{\affaddr}[1]{#1} 
\newcommand*{\Mark}[1]{\textsuperscript{#1}}
\newcommand{\forceindent}{\leavevmode{\parindent=1em\indent}}

\title{Recovering Metabolic Networks using A Novel Hyperlink Prediction Method\\}

\author{Muhan Zhang\Mark{1}, Zhicheng Cui\Mark{1}, Tolutola Oyetunde\Mark{2}, Yinjie Tang\Mark{2}, \and Yixin Chen\Mark{1}\\
\affaddr{\Mark{1}Department of Computer Science and Engineering, Washington University in St. Louis}\\
\affaddr{\Mark{2}Department of Energy, Environmental \& Chemical Engineering, Washington University in St. Louis}\\
\{muhan, z.cui, toyetunde, yinjie.tang\}@wustl.edu, chen@cse.wustl.edu
}

\begin{document}
\maketitle

\begin{abstract}
Studying metabolic networks is vital for many areas such as novel drugs and bio-fuels. For biologists, a key challenge is that many reactions are impractical or expensive to be found through experiments. Our task is to recover the missing reactions. By exploiting the problem structure, we model reaction recovery as a hyperlink prediction problem, where each reaction is regarded as a hyperlink connecting its participating vertices (metabolites). Different from the traditional link prediction problem where two nodes form a link, a hyperlink can involve an arbitrary number of nodes. Since the cardinality of a hyperlink is variable, existing classifiers based on a fixed number of input features become infeasible. Traditional methods, such as common neighbors and Katz index, are not applicable either, since they are restricted to pairwise similarities. In this paper, we propose a novel hyperlink prediction algorithm, called Matrix Boosting (MATBoost). MATBoost conducts inference jointly in the incidence space and adjacency space by performing an iterative completion-matching optimization. We carry out extensive experiments to show that MATBoost achieves state-of-the-art performance. For a metabolic network with 1805 metabolites and 2583 reactions, our algorithm can successfully recover nearly 200 reactions out of 400 missing reactions.
\end{abstract}

\section{Introduction}
Reconstructed metabolic networks are important tools for understanding the metabolic basis of human diseases, increasing the yield of biologically engineered systems, and discovering novel drug targets \cite{bordbar2014constraint}. Semi-automated procedures have been recently developed to reconstruct metabolic networks from annotated genome sequences \cite{thiele2010protocol}. However, these networks are often incomplete -- some reactions can be missing from them, which can severely impair their utility \cite{kumar2007optimization}. Thus, it is critical to develop accurate computational methods for completing metabolic networks. A method that can automatically recover missing reactions in metabolic networks is much desired.

Link prediction \cite{liben2007link,lu2011link} has been studied broadly in recent years \cite{singh2008relational,miller2009nonparametric,rendle2009bpr,menon2011link,cao2014collective,chen2015marginalized,ermics2015link,song2015top,wu2016modeling}. Existing methods can be grouped into three types: topological feature-based approaches, latent feature-based approaches, and supervised learning approaches. Popular criterions and methods include common neighbors, Jaccard coefficients \cite{liben2007link}, Katz index \cite{katz1953new}, and latent feature models \cite{miller2009nonparametric,rendle2009bpr,menon2011link} et al. However, these approaches are restricted to predicting pairwise relations. None of them are directly applicable to predicting hyperlinks. A hyperlink relaxes the assumption that only two nodes can participate in a link. Instead, an arbitrary number of nodes are allowed to jointly form a hyperlink. A network made up of hyperlinks is called a hypernetwork or hypergraph. 

Many complex networks are actually hypernetworks. Examples include biological networks, social networks with group relations, and citation networks. In metabolic networks, every reaction can be regarded as a hyperlink among its component metabolites. A hypernetwork can be naturally represented as an \textit{incidence matrix} $S$, where each column of $S$ is a hyperlink and each row represents a vertex. If vertex $i$ participates in hyperlink $j$, then $S_{ij}=1$, otherwise $S_{ij}=0$. In this paper, we model metabolic reaction recovery as a hyperlink prediction problem in hypernetworks. For metabolic networks, the $S$ matrix is also known as the \textit{Binary Stoichiometric Matrix}. Figure \ref{missingEx} shows an example.

Although hypernetworks are common in the real world, there are still limited research on hyperlink prediction. One great challenge lies in the variable cardinality of hyperlinks. Existing classifiers for link prediction are based on a fixed number of input features (e.g., concatenating features of vertex $a$ and features of vertex $b$). However, in hyperlink prediction, the number of vertices in a hyperlink is variable, which makes existing methods inapplicable. On the other hand, link prediction methods based on topological features, such as common neighbors, cannot be applied to hyperlink prediction either. The reason is that these measures are defined for pairs of nodes (links) instead of hyperlinks. A few basic methods for hyperlink prediction have poor performance on recovering metabolic networks. In summary, hyperlink prediction is an \textit{important, interesting, yet rarely touched challenging problem}.

In this paper, we introduce a novel hyperlink prediction algorithm for recovering metabolic networks. Our key observation is that a hyperlink $\mathbf{s}$ (a column vector in an incidence matrix) can be transformed into its equivalent matrix expression in the vertex adjacency space by $\mathbf{s}\mathbf{s}^{\rm{T}}$.
In this way, we are able to first infer the pairwise relationships between vertices leveraging existing link prediction methods, and then recover the missing hyperlinks through constrained optimization. Based upon this, we propose a two-step EM-like optimization method, which alternately performs a C step (completion) and an M step (matching) to find the best hyperlinks. We call the resulting algorithm Matrix Boosting (MATBoost). We compare our algorithm with six baseline methods on eight metabolic networks from five species. Experimental results demonstrate that our algorithm outperforms the state-of-the-art methods by a large margin on recovering missing metabolic reactions.

\section{Reaction Recovery as Hyperlink Prediction}
One significant difference between link prediction and hyperlink prediction is the total number of possible links / hyperlinks. For a network with $M$ vertices, the total number of links is $O(M^2)$, while the total number of hyperlinks will be $O(2^M)$. This prevents us from giving a score to every possible hyperlink like what we do in link prediction. Therefore, we consider a \textit{transductive inference} setting, where a set of test hyperlinks are given as candidates, and we only classify the given test examples instead of classifying all. This setting is natural in practical applications. For instance, in the reaction prediction task, we do not need to consider all $2^M$ possible reactions since most of them do not contain biological meanings. Instead, we restrict the candidate hyperlinks to be within a universal reaction pool. This pool is the set of all known metabolic reactions. Consequently, our task is to classify which reactions in the pool are the missing reactions for a specific metabolic network.
\subsection{Problem definition and notations}
We formally define the missing reaction recovery problem as follows: Given an $M\times N^S$ binary stoichiometric matrix $S$ ($M$ metabolites and $N^S$ reactions) of a metabolic network, we assume that several reactions are missing from it. We denote these missing columns by $\Delta S$. We also have a universal reaction pool, whose binary stoichiometric matrix is denoted by $U$. $U$ contains all known metabolic reactions, including both positive reactions ($\Delta S$) as well as negative reactions which do not belong to this metabolic network. Our task is to recover/predict as many missing reactions from $U$ as possible. The corresponding hyperlink prediction problem is parallel. Given a training incidence matrix $S$ (where some column hyperlinks $\Delta S$ are missing) and a testing incidence matrix $U$ (the candidate hyperlinks to select from), we aim to predict which columns in $U$ belong to $\Delta S$.

For an incidence matrix $S$, we can project it into the vertex adjacency space by $A = SS^{\rm T}$. $A$ is the weighted adjacency matrix where $A_{ij}$ is the number of hyperlinks that vertices $i$ and $j$ both participate in. In systems biology, $A$ is called the \textit{Compound Adjacency Matrix}. 

\begin{figure*}[tbp]
\centering
\includegraphics[width=0.85\textwidth]{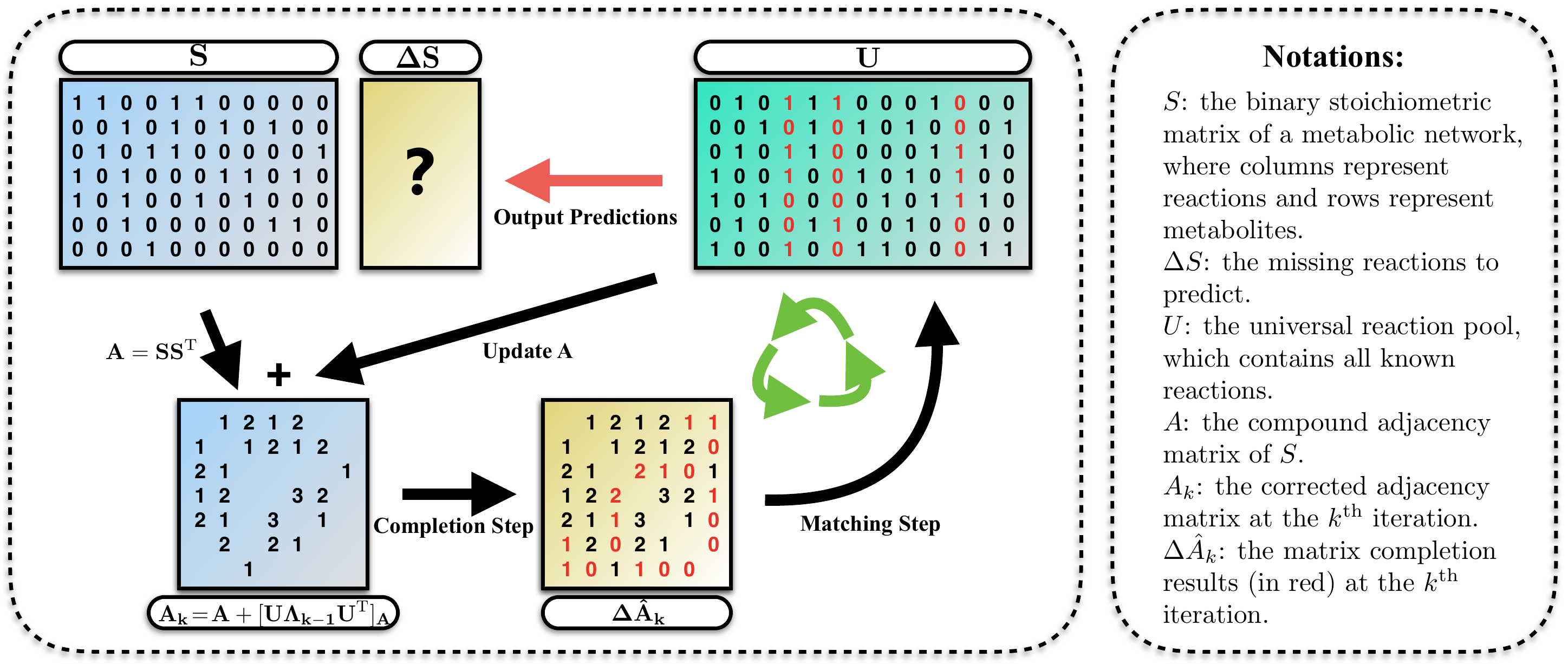}
\caption{\small An illustration of MATBoost. The training incidence matrix $S$ is first transformed to its adjacency matrix $A$. $A$ has many missing entries some of which are from $\Delta S$. The C (completion) step uses matrix factorization to complete the missing entries. The M (matching) step will search for the best matches from the candidate matrix $U$ by solving a constrained optimization. The matched hyperlinks will then be added to $A$ in order to correct $A$'s bias. This procedure is iterated several times to boost the performance.}
\label{missingEx}
\end{figure*}

\section{Matrix Boosting Algorithm}
We illustrate the proposed Matrix Boosting (MATBoost) algorithm in Figure \ref{missingEx}. The overall framework is as follows. 
 \begin{quote}
\textbf{Input}: Initial training adjacency matrix $A = SS^{\rm T}$, candidate hyperlinks $U$, initial zero prediction scores $\Lambda_0$, maximum iteration number $I_\text{max}$.\\
\textbf{Repeat from $k=1$}:
\begin{quote}
a) Update A, set: 
\end{quote}
\begin{equation}\label{bias2}
A_{k} = A + [U\Lambda_{k-1} U^{\rm T}]_{A}.
\end{equation}
\begin{quote}
b) The C (completion) step: 
Input $A_k$, output the matrix completion results $\Delta \hat{A}_k$.\\
c) The M (matching) step: 
Input $\Delta \hat{A}_k$, output the prediction scores $\Lambda_{k}$.\\
d) Stopping criterion:
If $||\Lambda_{k} - \Lambda_{k-1}||_F \geq ||\Lambda_{k-1} - \Lambda_{k-2} ||_F$ or  $k = I_\text{max}$, stop the algorithm; otherwise $k \leftarrow k+1$ and repeat from step a).
\end{quote}
\textbf{Output}: The final prediction scores: 
\begin{equation}
\bar{\Lambda} = \frac{1}{k-2} \sum_{i=1}^{k-2} \Lambda_i.
\end{equation}
\end{quote}
\vspace*{-0.3\baselineskip}
We now explain each step in the following subsections.


\subsection{Building block 1: decomposition of the adjacency matrix} 
Since direct inference in the incidence space is difficult (due to the variable cardinality of hyperlinks), we first project hyperlinks into their vertex adjacency space and infer their vertex interactions. Given the training hyperlinks $S$, we can calculate the training adjacency matrix by $A = SS^{\rm T}$. However, since $S$ is not complete (some columns $\Delta S$ are missing), the resulting $A$ is also incomplete.

Because $\Delta S$ is missing, the complete incidence matrix should be $\left[\!\! \begin{array}{cc}   S\!\! & \!\!\Delta S  \end{array}\!\!\right]$. We can calculate its adjacency matrix as follows:
\vspace*{-0.3\baselineskip}
\begin{align}\label{sst}
\left[\!\! \begin{array}{cc}   S \!\!& \!\!\Delta S  \end{array}\!\!\right]\left[\!\! \begin{array}{cc}   S \!\!& \!\!\Delta S  \end{array}\!\!\right]^{\rm T} &= SS^{\rm T} + \Delta S \Delta S^{\rm T}\nonumber\\ &= A + \Delta A,
\end{align}
where we define $\Delta A = \Delta S \Delta S^{\rm T}$. We notice that since hyperlinks $\Delta S$ are missing, the adjacency matrix $A$ also incurs a loss $\Delta A$. If we can accurately recover $\Delta A$, we will be able to infer $\Delta S$ consequently. 

A straightforward idea is to perform link prediction on $\Delta A$ using $A$. However, this will cause problems because $\Delta A$ may overlap with $A$ at some entries, while existing link prediction methods require that training links $A$ and testing links $\Delta A$ are fully separated.

To fix this issue, we first decompose $A+\Delta A$ as follows:
\begin{equation}\label{dec}
\begin{split}
A + \Delta A  =  \underbrace{A + [\Delta A]_{A}}_{A^+} + \underbrace{[\Delta A]_{\bar{A}}}_{\Delta A^-},
\end{split}
\end{equation}
where we have used $[X]_{A}$ to denote the operation that only keeps the entries of $X$ at $A$'s nonempty entries and mask all else, which is defined as:
\begin{equation}
\begin{split}
[X_{ij}]_{A} = \left\{ \begin{aligned} X_{ij},~~\text{if}~A_{ij} \neq 0\\
0,~~\text{if $A_{ij}= 0$}. \end{aligned} \right.
\end{split}
\end{equation}
$[X]_{\bar{A}}$ is conversely defined as keeping $X$ only at $A$'s empty entries.

Equation (\ref{dec}) decomposes the complete adjacency matrix into $A^+$ and $\Delta A^-$. $A^+$ is the training adjacency matrix $A$ plus the entries of $\Delta A$ that overlap with $A$. $\Delta A^-$ is $\Delta A$ discarding its entries that overlap with $A$. Such a decomposition ensures that $A^+$ and $\Delta A^-$ are isolated from each other. Our goal is thus to leverage $A^+$ to predict $\Delta A^-$, and subsequently, to recover $\Delta S$.

\vspace*{-0.7\baselineskip}
\subsection{Building block 2: the completion step}
We first assume that $A^+$ is known (which is not true because $\Delta A$ is unknown beforehand, we will show how to iteratively approximate it later). Then predicting $\Delta A^-$ can be regarded as a matrix completion \cite{candes2009exact} or weighted link prediction problem \cite{lu2010link}. Concretely, we aim to complete the empty entries of $A^+$, which we denote by $\Delta \hat{A}$. Ideally, the predicted $\Delta \hat{A}$ will resemble $\Delta A^-$ by having large values at $\Delta A^-$'s nonempty entries and near-zero values at other places. 

In this paper, we adopt matrix factorization \cite{koren2009matrix} to infer $\Delta A^-$ because it is fast and accurate.
The optimization problem is as follows:
\begin{equation}\label{mftrain}
\begin{split}
&\mathop{\text{minimize}}_\Theta ~~ \sum \limits_{i<j, A^+_{ij}>0} ||A^+_{ij} - y_{ij}||_2^2 + \gamma\mathcal{R}(\Theta),\\
&~~\text{where}~~ y_{ij} = w_0 + w_i + w_j + \sum_{f=1}^k v_{if} v_{jf},\\
&~~~~ \Theta ~\text{is the set of all parameters $w_0, w_i, w_j, v_{if}, v_{jf}$}.
\end{split}
\end{equation}
The optimization only uses half of the training entries since the adjacency matrix $A^+$ is symmetric. After training, we can complete the empty entries by:
\begin{align}\label{mfpredict}
\Delta\hat{A}_{ij} = 
\left\{ \begin{aligned}& w_0 + w_i + w_j + \mathbf{v}_i^{\rm T} \mathbf{v}_j,~~\text{if}~A^+_{ij}=0,\\
&0, ~~\text{if}~A^+_{ij}\neq0.
\end{aligned} \right.
\end{align}

\subsection{Building block 3: the matching step}
In the matching step, we aim to select a a batch of hyperlinks from $U$ that best match the predicted $\Delta \hat{A}$. Firstly, we define a diagonal matrix $\Lambda$ containing the selection indicators $\lambda_i$ for the $i^\text{th}$ hyperlink in $U$. $\lambda_i$ is 1 if hyperlink $i$ is selected and 0 otherwise. 

Now consider the optimization problem as follows:
\begin{equation}\label{opt1}
\begin{split}
&\mathop{\text{minimize}}_\Lambda ~~ ||[U\Lambda U^{\rm T}]_{\bar{A}} - \Delta A^-||_{F}^2, \\
&\text{subject to} ~~\lambda_i \in \{0,1\}, ~ i=1~...~N^U, \\
&\text{where} ~~ \Lambda = \left[ \begin{array}{cccc}   \lambda_1 &  &  &\\  & \lambda_2 & & \\  & &... & \\  & & & \lambda_{N^U} \end{array}\right],\\
& ~~~~~~~~\text{$N^U$ is the number of columns in $U$,}\\
&~~~~~~~~\text{$||\cdot||_F$ is the Frobenius norm.}
\end{split}
\end{equation}

\begin{theorem}
The optimization problem (\ref{opt1}) is an integer least square problem with the ground-truth indicators $\Lambda^*$ being a global minimum of it. (proof in the supplemental material)
\end{theorem}

In practice, we will substitute $\Delta \hat{A}$ for $\Delta A^-$ in (\ref{opt1}) to solve $\Lambda$. However, integer least square (ILSQ) optimization is NP hard. When $N^U$ is large (usually thousands in metabolic networks), it is generally intractable. Therefore, we propose a constrained lasso regression instead. The optimization problem is as follows:
\vspace*{-0.5\baselineskip}
\begin{equation}\label{lasso}
\begin{split}
&\mathop{\text{minimize}}_\Lambda ~~ ||[U\Lambda U^{\rm T}]_{\bar{A}} - \Delta \hat{A}||_{F}^2 + \alpha \sum_{i=1}^{N^U} |\lambda_i|, \\
&\text{subject to} ~~0\leq \lambda_i \leq1, ~ i=1~...~N^U, \\
&\text{where} ~~ \Lambda = \left[ \begin{array}{cccc}   \lambda_1 &  &  &\\  & \lambda_2 & & \\  & &... & \\  & & & \lambda_{N^U} \end{array}\right],\\
\end{split}
\end{equation}

The above optimization relaxes the integer $\lambda_i$ to be continuous within $[0,1]$. Because of the limited number of hyperlinks in U that contribute to $\Delta A^-$, we add an lasso ($L_1$) penalty to preserve the sparsity of solutions. We solve the above optimization using subgradient methods. After getting the prediction scores $\lambda_i$, we can rank the hyperlinks in $U$ and select the top ones as our predictions of the missing hyperlinks. 

\subsection{Combine all: boosting via iterative optimization}
In the completion step, we have assumed that $A^+ = A + [\Delta A]_A$ is known in order to separate $\Delta A^-$ from $A^+$ and infer $\Delta A^-$. Notice that $A^+$ is essentially the training adjacency matrix $A$ plus the overlapping entries of the missing hyperlinks $\Delta S\Delta S^{\rm{T}}$. Thus, to approximate the unknown $A^+$, we propose to iteratively add the predicted missing hyperlinks $U\Lambda U^{\rm T}$ into $A$ ($\Lambda$ is real matrix now) to remedy its bias from the real $A^+$.

Given the last iteration's prediction scores $\Lambda_{k-1}$, we first update $A_k$ by adding to $A$ the predicted hyperlinks $[U\Lambda_{k-1} U^{\rm T}]_{A}$ from the last iteration. After getting $A_k$, which is an approximation of $A^+$, the algorithm performs the $C$ step and $M$ step subsequently. MATBoost performs this procedure in an iterative fashion and outputs the average prediction scores in order to boost the performance. It is further equipped with an early stopping criterion which we will explain in the next section. 

\section{Algorithm Analysis}
We provide two different views of the MATBoost algorithm in this section. First we define some notations:
\begin{definition}
We define the C (completion) step (\ref{mftrain}) and (\ref{mfpredict}) as a map $T_c:~\mathcal{A} \rightarrow \mathcal{A}$, where $\mathcal{A}$ is the space of adjacency matrices $A$. $T_c$ takes an incomplete adjacency matrix $A_k$ as input and outputs the predictions $\Delta \hat{A}_k$.
\end{definition}

\begin{definition}
We define the M (matching) step (\ref{lasso}) as a map $T_m:~\mathcal{A} \rightarrow \mathcal{L}$~, where $\mathcal{L}$ is the space of diagonal prediction score matrices $\Lambda$. $T_m$ takes the predictions $\Delta \hat{A}_k$ as input and outputs the prediction scores $\Lambda_{k}$.


\end{definition}

\subsection{A coordinate descent view}
MATBoost may be viewed as an example of coordinate descent algorithms. In the C step, we fix $\Lambda$ and update $\Delta \hat{A}$. In the M step, we fix $\Delta \hat{A}$ and update $\Lambda$. Concretely, MATBoost can be viewed as an alteration between the two steps:
\begin{quote}
C (completion) step:\\
\forceindent Fix $\Lambda_{k-1}$, update $\Delta \hat{A}$ by
$\Delta \hat{A}_k \!\!=\!\! T_c(A+[U\Lambda_{k-1}U^{\rm{T}}]_A)$.\\
M (matching) step:\\
\forceindent Fix $\Delta \hat{A}_k$, update $\Lambda$ by $\Lambda_k \!=\! T_m(\Delta \hat{A}_k)$.\\
\end{quote}

\vspace*{-1.5\baselineskip}
\subsection{A fixed-point view}
Here we discuss the convergence properties and the stopping criterion of MATBoost under a fixed-point view.
\begin{definition}
We define one iteration of both C and M steps as a map $T_{cm}:~\mathcal{L} \rightarrow \mathcal{L}$~. $T_{cm}$ takes the $\Lambda_{k-1}$ from the last iteration as input and outputs the new prediction scores $\Lambda_{k}$. Equivalently, $T_{cm}(\Lambda_{k-1}) = T_m(T_c(A+[U\Lambda_{k-1}U^{\rm{T}}]_A))$.
\end{definition}

Having defined $T_{cm}$, we find that the MATBoost algorithm can be expressed as a fixed-point iteration algorithm \cite{khamsi2011introduction}, where the algorithm iteratively takes the intermediate result $\Lambda_{k} = T_{cm}(\Lambda_{k-1})$ as input to get $\Lambda_{k+1}$, $\Lambda_{k+2}$... until the stopping criterion is satisfied. 


Now we introduce the Banach fixed-point theorem.

\begin{theorem}\label{t1}
(Banach fixed-point theorem \cite{palais2007simple}) If $T_{cm}$ is a contraction mapping, then MATBoost will iteratively converge to a fixed-point $\Lambda^*$, which satisfies $\Lambda^* = T_{cm}(\Lambda^*)$.
\end{theorem}

\begin{theorem}\label{t2}
If $T_{cm}$ is a contraction mapping defined on the metric space $(\mathcal{L},d)$ and $d(\Lambda_x,\Lambda_y) = ||\Lambda_x-\Lambda_y||_F$, then the intermediate results $\Lambda_1,...,\Lambda_{k-2}$ of MATBoost satisfy the following inequality for any $k>2$:
\begin{equation}\label{c1}
||\Lambda_{k-1}-\Lambda_{k-2}||_F < ||\Lambda_{l-1}-\Lambda_{l-2}||_F, ~~\forall~ 2\leq l<k.
\end{equation}
Besides, $\Lambda_{k-2}$ is the closest point to $\Lambda^*$ among $\Lambda_0, ..., \Lambda_{k-2}$. (proof in the supplemental material)
\end{theorem}

From Theorem \ref{t2} we see that $||\Lambda_{k} - \Lambda_{k-1}||_F < ||\Lambda_{k-1} - \Lambda_{k-2} ||_F$ for all $k$ is a necessary condition for $T_{cm}$ being a contraction mapping on $(\mathcal{L},d)$. Recall that the stopping criterion in MATBoost is $||\Lambda_{k} - \Lambda_{k-1}||_F \geq ||\Lambda_{k-1} - \Lambda_{k-2} ||_F$. We can see that it is essentially tracking this necessary condition. Once it is broken, the fixed-point iteration will stop converging. And the algorithm will stop the iteration earlier by outputting the average scores before iteration $k-1$.

Since $||T_{cm}(\Lambda^*)-\Lambda^*||_F=0$, we may also regard $||\Lambda_{k} - \Lambda_{k-1}||_F=||T_{cm}(\Lambda_{k-1}) - \Lambda_{k-1}||_F$ as a measure of how close the iteration $k-1$ is to convergence.

In practice, we find that $||\Lambda_{k}-\Lambda_{k-1}||_F$ typically drops sharply in the first iteration, then keeps decreasing for a few rounds, and then starts to fluctuate around a small value. Instead of using the last $\Lambda_{k-2}$, we find that using the mean scores of $\Lambda_1$ to $\Lambda_{k-2}$ usually achieves higher performance. We have Theorem \ref{t4} to explain it.

\begin{theorem}\label{t4}
The MATBoost algorithm is an ensemble of $k-2$ weak learners of $T_{cm}$ with different input matrices $A + [U\Lambda_i U^{\rm T}]_{A}$, $i$ from $0$ to $k-3$.
\end{theorem}
\begin{proof}
It is easy to prove by noticing that $\Lambda_1$ to $\Lambda_{k-2}$ are just the outputs of $T_{cm}(\Lambda_0)$ to $T_{cm}(\Lambda_{k-3})$. Thus the final prediction scores can be written as:
\begin{equation}
\bar{\Lambda} = \frac{1}{k-2} \sum_{i=1}^{k-2} \Lambda_i = \frac{1}{k-2} \sum_{i=0}^{k-3} T_{cm}(\Lambda_i).
\end{equation}
The theorem follows.
\end{proof}
\vspace*{-0.3\baselineskip}
\noindent \textbf{Remark}: We may regard $A + [U\Lambda_i U^{\rm T}]_{A}$ as noisy samplings of $A^+$. Therefore, MATBoost can benefit from ensemble learning which is expected to improve performance.

\section{Related Work}

Although hyperlinks are common in real world and can be used to model multi-way relationships, currently there are still limited research on hyperlink prediction. Xu et al. \cite{xu2013hyperlink} proposed a supervised HPLSF framework to predict hyperlinks in social networks. To deal with the variable number of features, HPLSF uses their entropy score as a fixed-length feature for training a classification model. To our best knowledge, this is the only algorithm that is specifically designed for hyperlink prediction in arbitrary-cardinality hypernetworks.


Nevertheless, learning with hypergraphs as a special data structure is popular in the machine learning community, e.g., semi-supervised learning with hypergraph regularization \cite{yu2012adaptive,tian2009hypergraph}, modeling label correlations via hypernetworks in multi-label learning \cite{sun2008hypergraph}, and modeling communities to improve recommender systems \cite{bu2010music}. Zhou et al. \cite{zhou2006learning} studied spectral clustering in hypergraphs. They generalized the normalized cut \cite{shi2000normalized} algorithm to hypergraph clustering and proposed a hypergraph Laplacian. They also adapted their previous semi-supervised learning algorithm using graph regularization \cite{zhou2005learning} to hypergraph vertex classification. These research mainly aim to improve the learning performance on nodes by leveraging their hyperlink relations. However, none of them focuses on predicting the hyperlink relations. 

\begin{figure*}[htp]
\centering
\subfloat[iJO1366 dataset.]{\includegraphics[width=0.24\textwidth]{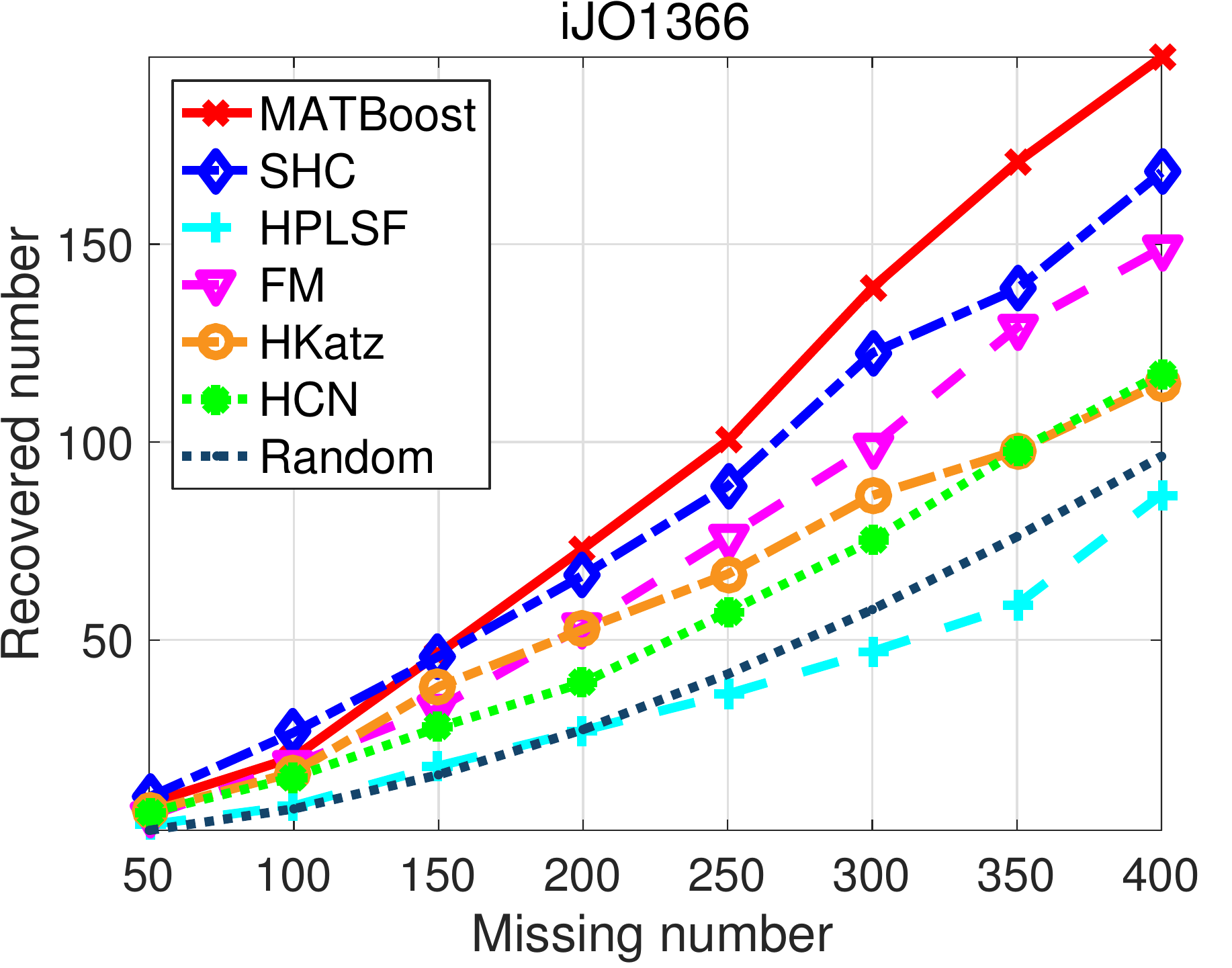}}
\subfloat[iAF1260b dataset.]{\includegraphics[width=0.24\textwidth]{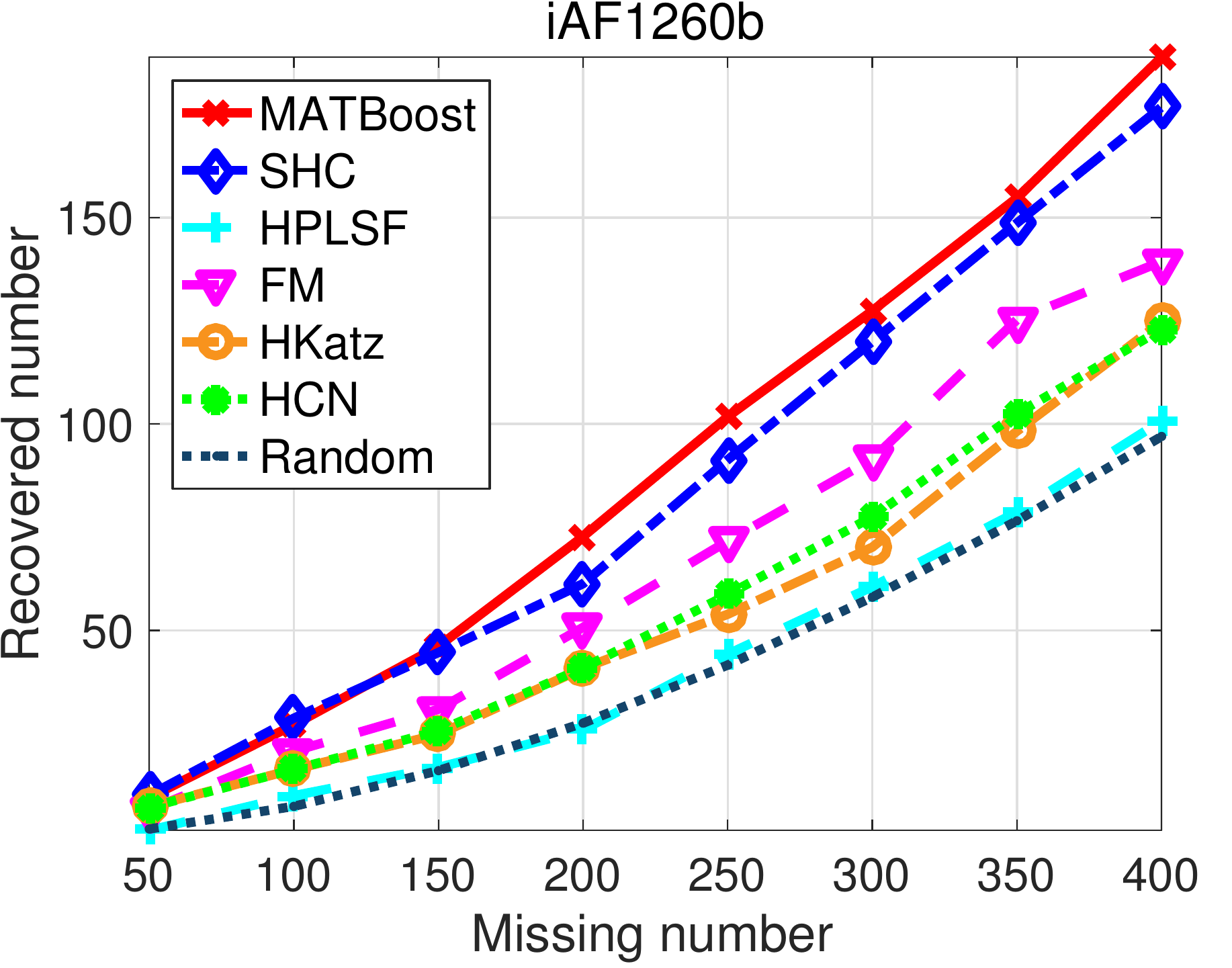}}
\subfloat[e\_coli\_core dataset.]{\includegraphics[width=0.24\textwidth]{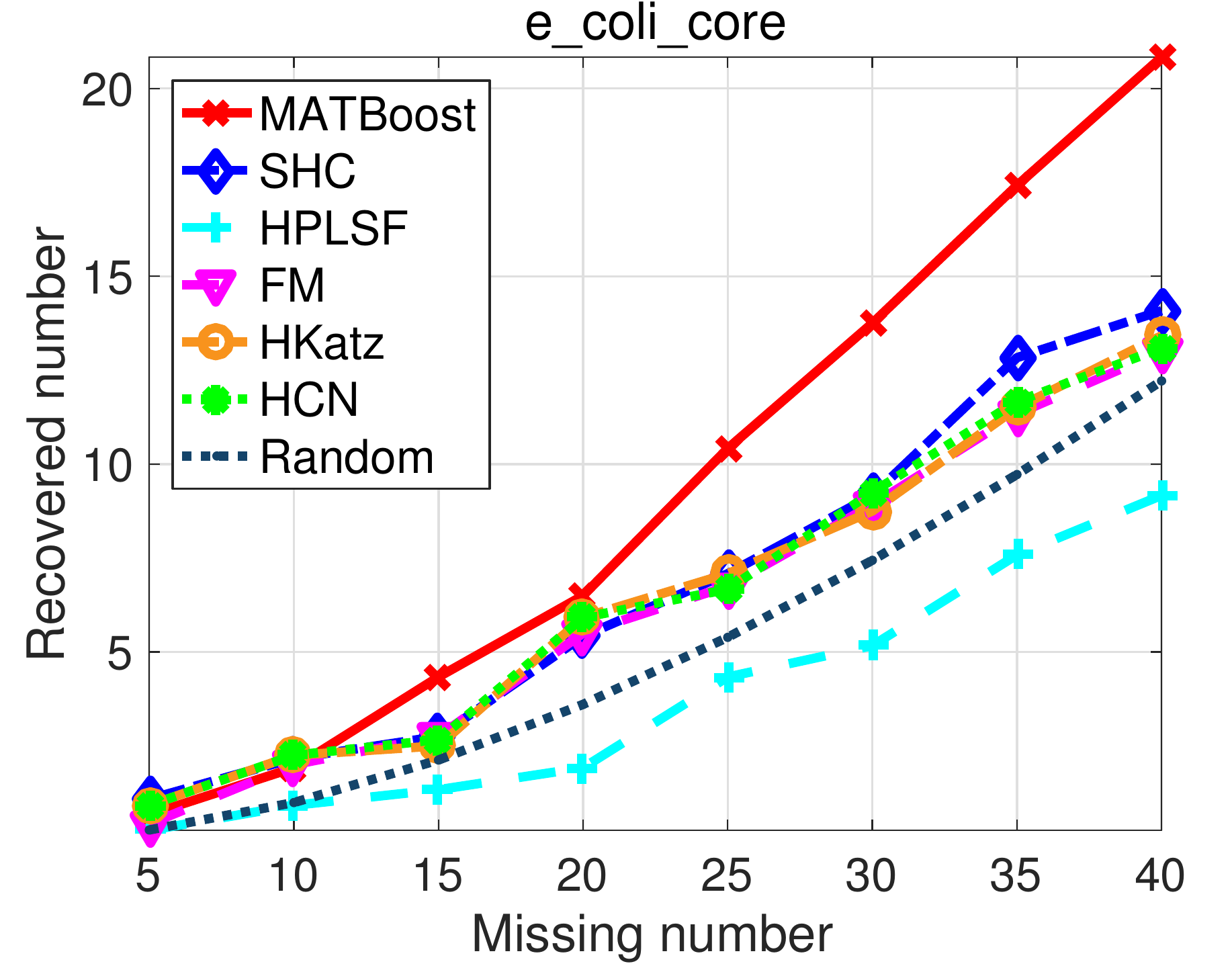}}
\subfloat[RECON1 dataset.]{\includegraphics[width=0.24\textwidth]{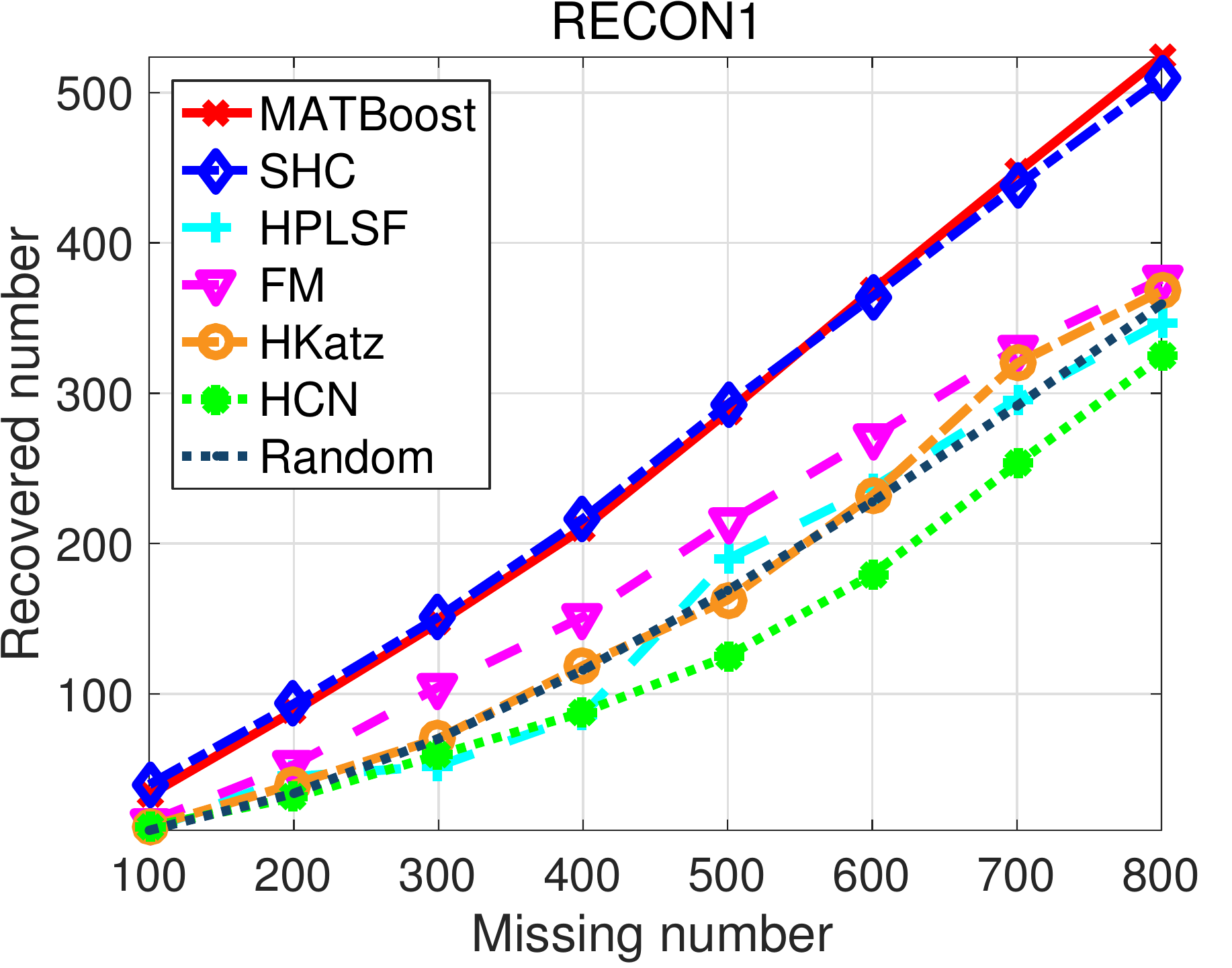}}

\centering
\subfloat[iAT\_PLT\_636 dataset.]{\includegraphics[width=0.24\textwidth]{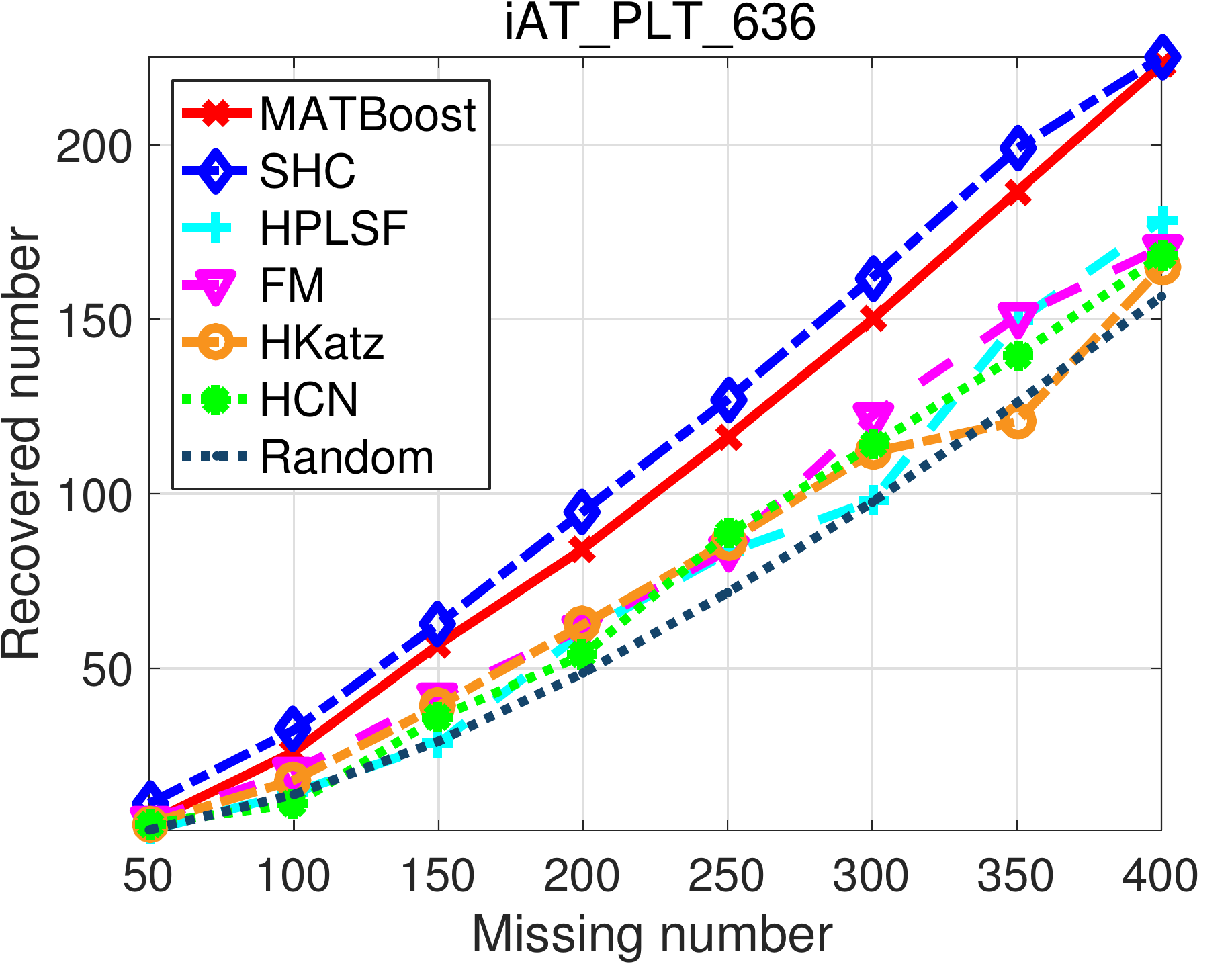}}
\subfloat[iAF692 dataset.]{\includegraphics[width=0.24\textwidth]{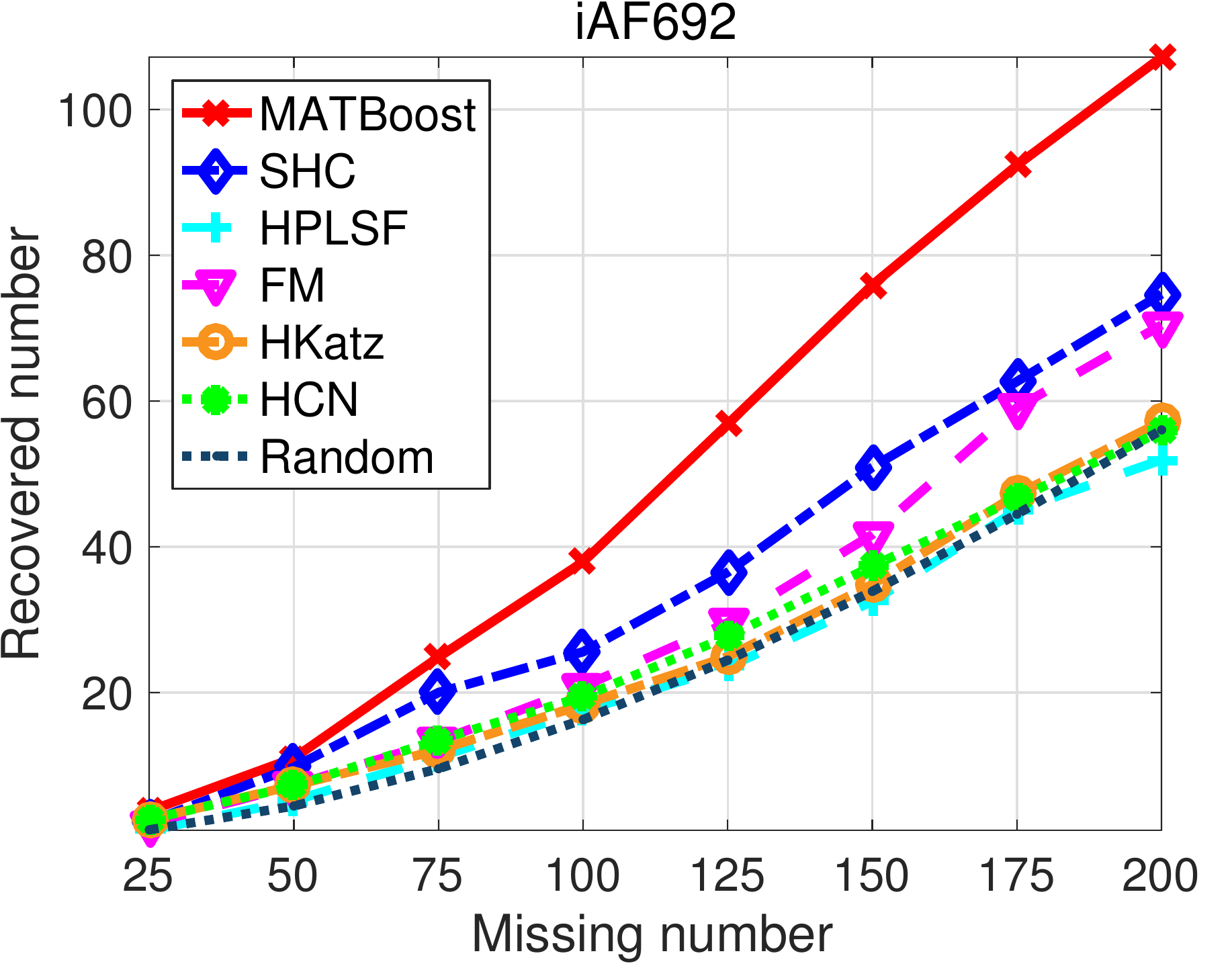}}
\subfloat[iHN637 dataset.]{\includegraphics[width=0.24\textwidth]{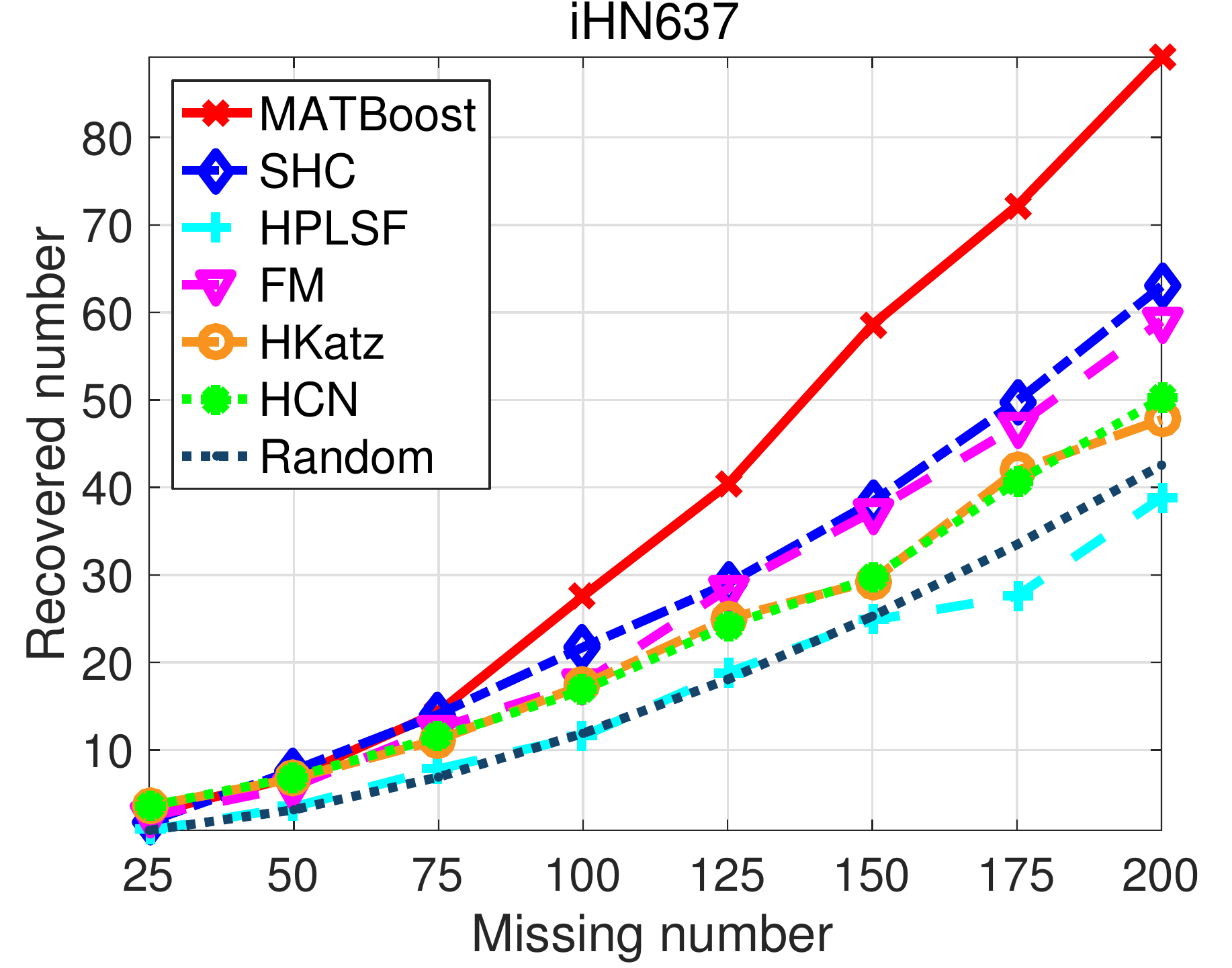}}
\subfloat[iIT341 dataset.]{\includegraphics[width=0.24\textwidth]{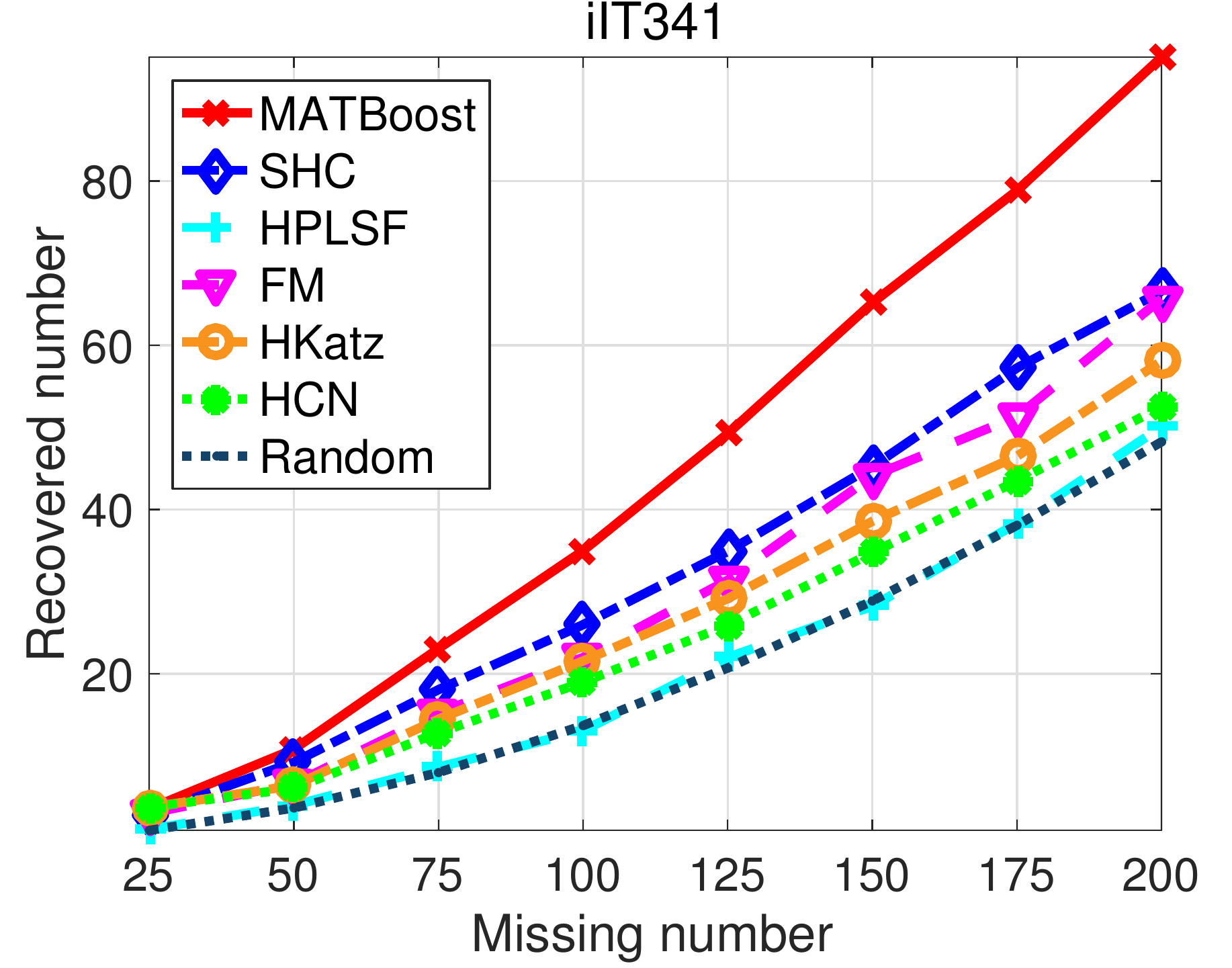}}
\caption{\small Number of recovered reactions for different numbers of missing reactions.}
\label{RMfig}
\end{figure*}

\section{Experimental Results}
To demonstrate the effectiveness of the proposed Matrix Boosting (MATBoost) algorithm in reaction recovery tasks, we compare it with six baseline methods on eight metabolic networks. The datasets and source code in this paper are available online: https://github.com/ufo0010/HyperLinkPrediction.

\subsection{Baselline algorithms}
We compare the proposed Matrix Boosting algorithm with the following six baseline algorithms:

\noindent \textbf{SHC} (Spectral Hypergraph Clustering) \cite{zhou2006learning} is a state-of-the-art hypergraph learning algorithm. SHC outputs classification scores by $f = (I-\xi\Theta)^{-1}y$. The hyperparameter $\xi$ is determined by searching over the grid \{0.01,0.1,0.5,0.99,1\} using five-fold cross validation. SHC is originally designed to classify hypergraph vertices leveraging their hyperlink relations. Here we transpose the incidence matrices to change each vertex into a hyperlink and each hyperlink into a vertex. Thus SHC can be used for hyperlink prediction.

\noindent \textbf{HPLSF} \cite{xu2013hyperlink} is a hyperlink prediction method using supervised learning. It calculates an entropy score among vertices along each latent feature dimension in order to get a fixed-length feature input. We train a logistic regression model on these entropy features in order to output prediction scores. 

\noindent \textbf{FM} (Factorization Machine) \cite{rendle2010factorization,rendle2012factorization} is a flexible factorization model. We use the classification function of FM, where columns of the training incidence matrix are directly used as input features to the model.

\noindent \textbf{HKatz} is a generalization of Katz index \cite{katz1953new} for hyperlinks. Concretely, a hyperlink containing $m$ vertices will have $m(m-1)/2$ pairwise Katz indices. We calculate their average as the HKatz (Hypernetwork Katz) index of the hyperlink. The parameter $\beta$ is determined by searching over \{0.001,0.005,0.01,0.1,0.5\} using five-fold cross validation.

\noindent \textbf{HCN} is a generalization of common neighbors \cite{liben2007link} for hyperlinks. For a hyperlink containing $m$ vertices, the average of the $m(m-1)/2$ pairwise common neighbors is calculated as the HCN (Hypernetwork Common Neighbors) score for this hyperlink.

\noindent \textbf{Random}: a theoretical baseline for comparing algorithms' performance against random. It is equal to assigning random scores between [0,1] to all testing hyperlinks.

In MATBoost and FM, libFM \cite{rendle2012factorization} is adopted as the default matrix factorization software. The optimization scheme is MCMC, which has advantages of high accuracy and automatic regularization. The latent factor number $k$ is set to 8 (default) in both algorithms. The lasso hyperparameter $\alpha$ and the maximum iteration number of MATBoost is set to 0.1 and 10 respectively.

\subsection{Evaluation and implementation}
We evaluate the reaction recovery performance based on their prediction scores. We adopt AUC (area under the ROC curve) as one measure. We also evaluate performance based on how many missing reactions can be successfully recovered if we select top $n$ reactions in $U$ as predictions ($n$ is the number of missing reactions). We call the second measure ``Recovered number". Compared to AUC, this measure only focuses on top predictions and can more directly reflect how well hyperlink prediction algorithms perform in practice.

We implement our algorithms mainly in MATLAB. All experiments are done on a twelve-core Intel Xeon Linux server. All experiments are repeated twelve times and the average results are presented.

\subsection{Datasets and results}

We download all 11893 reactions from BIGG\footnote{http://bigg.ucsd.edu} to build a universal reaction pool. These reactions are collected from 79 metabolic networks of various organisms. We conduct experiments on eight metabolic networks from five species: E. coli, H. sapiens, M. barkeri, Cl. ljungdahlii, and H. pylori. For each metabolic network, we filter out the universal reactions which contain exotic metabolites from the pool. We also filter out the network's existing reactions from the pool. Table \ref{tab} shows the statistics of the eight datasets.
\begin{table}[htbp]
\centering
\small
\resizebox{0.47\textwidth}{!}{
\begin{tabular}{l|c|c|c|c}
    \hline
    \hline
    Dataset & Species & Vertices & Hyperlinks & Neg. Hyperlinks \\
    \hline
    \hline
    (a) iJO1366 & E. coli & 1805 & 2583 & 1253 \\
    \hline
    (b) iAF1260b & E. coli & 1668 & 2388 & 1242 \\
    \hline
    (c) e\_coli\_core & E. coli & 72 & 95 & 91 \\
    \hline
    (d) RECON1 & H. sapiens & 2766 & 3742 & 973 \\
    \hline
    (e) iAT\_PLT\_636 & H. sapiens & 738 & 1008 & 621 \\
    \hline
    (f) iAF692 & M. barkeri & 628 & 690 & 513 \\
    \hline
    (g) iHN637 & Cl. Ljungdahlii & 698 & 785 & 740 \\
    \hline
    (h) iIT341 & H. pylori & 485 & 554 & 629 \\
    \hline
 \end{tabular}}
 \caption{Statistics of the eight datasets.}
 \label{tab}
\end{table}

For each dataset, we randomly delete [5:5:40], [25:25:200], [50:50:400] or [100:100:800] (according to network size) reactions as missing hyperlinks, and keep the remaining ones as training data. Due to the space limit, we only present the results of ``Recovered number''. The AUC results are in the supplemental material. 

As shown in Figure \ref{RMfig}, MATBoost generally achieves higher performance than other baselines. We observe that MATBoost recovers a significantly larger number of reactions than other baselines in datasets (a) -- (c) and (f) -- (h), and achieves competitive results with SHC in datasets (d) and (e). \textbf{This is because}: 1) MATBoost avoids directly performing inference in the incidence space that has size $O(2^M)$ by first inferring the adjacency matrix. This transforms an $O(2^M)$ problem into an $O(M^2)$ problem, which greatly reduces the problem size and eases the problem; 2) MATBoost leverages the powerful matrix factorization technique to perform inference in the adjacency space, which is able to well recover the unseen compound interactions given by $\Delta A^-$; 3) MATBoost infers the incidence matrix by solving a constrained optimization problem, which recovers the reactions that best match the predicted compound interactions in adjacency space; 3) MATBoost takes advantage of a fixed-point iteration to boost performance and benefits from ensemble averaging. In all eight datasets, MATBoost is able to successfully recover over or nearly half of the missing reactions with the increasing of missing number. This high recall in predicting top reactions is very important because biologists usually can only verify top predictions due to resource and time limits. 

In addition, we observe that SHC generally ranks second. This is an interesting result indicating that vertices and hyperlinks in a hypernetwork are sometimes interchangeable (by transposing the incidence matrix). We find that HKatz and HCN do not perform well. The phenomenon suggests that simply generalizing topological indices of links to hyperlinks by averaging may not be able to reflect the likelihood of a hyperlink. In terms of time complexity, MATBoost is also efficient. It takes only a few minutes to finish on the largest task.


\vspace*{-1\baselineskip}

\section{Conclusions}
In this paper, we consider the important task of predicting missing metabolic reactions for metabolic networks. By exploiting the problem structure, we model it as a hyperlink prediction problem. Hyperlink prediction is an interesting and challenging problem, however, relevant research are still limited. We proposed a novel algorithm for hyperlink prediction, called Matrix Boosting (MATBoost). It leverages the problem structure and implements an alternating C-M optimization framework in order to perform inference jointly in the incidence space and adjacency space. We have provided extensive evaluation results to demonstrate the unmatched performance of MATBoost. We compared our algorithm with six baselines. Experimental results on eight metabolic networks demonstrated that our MATBoost algorithm could recover over half of the missing reactions, which outperformed all existing methods by a large margin. 

Hyperlink prediction may also be useful in many other application domains such as collaboration, authorship, and social networks. We will study them in our future work.
\bibliographystyle{aaai}
\bibliography{references}

\end{document}